\documentclass[12pt,a4paper]{article}
\usepackage{verbatim}
\usepackage[utf8]{inputenc} 
\usepackage[english]{babel} 
\usepackage{anysize} 
\usepackage{indentfirst}
\usepackage{amsmath}
\usepackage{amsthm} 
\usepackage{amsfonts}
\usepackage{amssymb}
\usepackage{graphicx}
\usepackage{enumitem}

\usepackage{subfig} 
\usepackage{xcolor}

\usepackage{fancyhdr}
\usepackage{calligra}

\newtheorem{theorem}{Theorem}[section]

\newtheorem{definition}[theorem]{Definition}

\title{On Extensions of Quasi-Overlap and Quasi-Grouping Functions Defined on Bounded Lattices}
\author{Ana Shirley Monteiro \and Regivan Santiago \and Benjamín Bedregal \and Juscelino Araújo \and Eduardo Palmeira}

\date{}

\begin{document}
\maketitle

\section{Introduction}

In this paper, we propose a method of extending quasi-overlap and grouping functions defined on a sublattice $ M $ of a bounded lattice $ L $ to this lattice considering a more general version of sublattice definition, introduced by Palmeira and Bedregral in \cite{palmeira2012extension}.

\section{Extensions of Quasi-Overlap Functions}

In what follows we recall the notion of \emph{quasi-overlap function}.

\begin{definition}[Quasi-overlap function,  \cite{Ovlp&QOvlp_Paiva2021}] 	
	\label{def:quasioverlap}
	A binary operator $ O: L^2 \to L $ is said to be a \emph{quasi-overlap function} on $ L $ if, for any $ x, y, z \in L $, the following conditions hold:
	\begin{enumerate}[leftmargin=*,label={\rm($\mathbf{QO\arabic*}$)}]
		\item \label{item_QO1} commutativity: $ O(x,y) = O(y,x) $;
		\item \label{item_QO2} $ O(x,y) = 0_L $ iff $ x = 0_L $ or $ y=0_L $;
		\item \label{item_QO3} $ O(x,y)=1_L $ iff $ x = 1_L $ and $ y = 1_L $;
		\item \label{item_QO4} increasingness: $ O(x,y) \leq_L O(x,z) $ whenever $ y \leq_L z $.
	\end{enumerate}
\end{definition}

Suppose there exists a retraction $ r : L \to M $ with pseudo-inverse $ s : M \to L $. If $ O $ is a quasi-overlap function on $ M $, we want to provide a quasi-overlap function $ O^E $ on $ L $ such that 
\begin{center}
	$O^{E} \circ (s \times s) = s \circ O$
\end{center} 
In this context we will say that \emph{$O^{E}$ extends $ O $ from $ M $ to $ L $}.

\begin{theorem}
	\label{theo:extensofquasiover}
	If $O$ is a quasi-overlap on $M$ and $r: L \to M$ is a retraction with pseudo-inverse $s: M \to L$ such that
	\begin{align}
		\label{eq:extensofquas1}
		r(x) = 0_M  &\Leftrightarrow  x = 0_L;\\
		\label{eq:extensofquas2}
		r(x) = 1_M  &\Leftrightarrow  x = 1_L;
	\end{align}
	then $O^{E}: L^2 \to L$, defined by 
	\begin{equation}
		\label{eq:extensofquas3}
		O^{E}(x, y) = s(O(r(x), r(y))),
	\end{equation}
	is a quasi-overlap function that extends $ O $ from $ M $ to $ L $.
\end{theorem}
\begin{proof} For all $(x,y) \in M \times M$, we have:
\begin{align*}
\left[ O^{E} \circ (s \times s)\right] (x,y) &= O^{E}(s(x),s(y)) = s\left(O\left(r(s(x)),r(s(y))\right)\right) = s\left(O\left(x,y\right)\right) = \left( s \circ O\right) (x,y). & 
	\end{align*}
	Thus, $O^{E}$ extends $ O $ from $ M $ to $ L $. Furthermore, the following properties hold:
	
	\begin{enumerate}[leftmargin=*,label={\rm(QO\arabic*)}]	
	\item $ O^E $ is commutative once $ O $ is commutative;
		
	\item For all $(x,y) \in L^2$, we have:
	\begin{align*}
	O^E(x,y) = 0_L & \Leftrightarrow s\left(O\left(r(x),r(y)\right)\right) = 0_L & [\text{By Eq.} \eqref{eq:extensofquas3}.]\\
			& \Leftrightarrow O\left(r(x),r(y)\right) = 0_M & [\text{By Eq.\eqref{eq:extensofquas1} and $ r \circ s = Id_M $.}]\\
			& \Leftrightarrow r(x) = 0_M \   \text{or} \ r(y) = 0_M & [\text{By \ref{item_QO2}}.]\\
			& \Leftrightarrow x = 0_L \   \text{or} \ y = 0_L. & [\text{By Eq.}  \eqref{eq:extensofquas1}.]
		\end{align*}
		
		\item Analogous to (QO2).
				
		\item $ O^E $ is an increasing function since $ r $, $ s $ and $ O $ are increasing.
		
		%
	\end{enumerate}
\end{proof}

\section{Extensions of Quasi-Grouping Functions}

In what follows we remind the notion of \emph{quasi-grouping function}. 

\begin{definition}[Quasi-grouping function,\cite{QiaoQuasiGroup2022}] 
	\label{def:quasigroupping}
	A binary operator $ G: L^2 \to L $ is said to be a \emph{quasi-grouping function} on $ L $ if, for any $ x, y, z \in L $, the following conditions hold:
	\begin{enumerate}[leftmargin=*,label={\rm($\mathbf{QG\arabic*}$)}]
		\item \label{item:QG1} commutativity: $ G(x,y) = G(y,x) $;
		\item \label{item:QG2} $ G(x,y) = 0_L $ iff $ x = 0_L $ and $ y=0_L $;
		\item \label{item:QG3} $ G(x,y)=1_L $ iff $ x = 1_L $ or $ y = 1_L $;
		\item \label{item:QG4} increasingness: $ G(x,y) \leq_L G(x,z) $ whenever $ y \leq_L z $.
	\end{enumerate}
\end{definition}

\medskip 

\begin{theorem}
\label{theo:extensofqgroup}
If $G$ is a quasi-grouping on $M$ and $r: L \to M$ is a retraction with pseudo-inverse $s: M \to L$ satisfying Eq. \eqref{eq:extensofquas1} and Eq. \eqref{eq:extensofquas2}, then $G^{E}: L^2 \to L$, defined by 
\begin{equation}
\label{eq:extensofqgroup}
		G^{E}(x, y) = s(G(r(x), r(y))),
	\end{equation}
	is a quasi-grouping function that extends $ G $ from $ M $ to $ L $.
\end{theorem}

\begin{proof} For all $(x,y) \in M \times M$, we have:
\begin{align*}
\left[ G^{E} \circ (s \times s)\right] (x,y) & = G^{E}(s(x),s(y)) = s\left(G\left(r(s(x)),r(s(y))\right)\right) = s\left(G\left(x,y\right)\right) = \left( s \circ G\right) (x,y). & 
\end{align*}
Thus, $G^{E}$ extends $ G $ from $ M $ to $ L $. In addition, the following properties hold:
	\begin{enumerate}[leftmargin=*,label={\rm(QG\arabic*)}]
		\item $ G^E $ is commutative, since $ G $ is commutative;
		
		\item For all $(x,y) \in L^2$, we have:
		\begin{align*}
			G^E(x,y) = 0_L & \Leftrightarrow s\left(G\left(r(x),r(y)\right)\right) = 0_L & [\text{By Eq.} \eqref{eq:extensofqgroup}.]\\
			& \Leftrightarrow G\left(r(x),r(y)\right) = 0_M & [\text{By Eq.\eqref{eq:extensofquas1} and $ r \circ s = Id_M $.}]\\
			& \Leftrightarrow r(x) = 0_M \   \text{and} \ r(y) = 0_M & [\text{By \ref{item:QG2}.}]\\
			& \Leftrightarrow x = 0_L \  \text{and} \ y = 0_L & [\text{By Eq.}  \eqref{eq:extensofquas1}.]
		\end{align*}
		
		\item Analogous to (QG2).
				
		\item $ G^E $ is an increasing function since $ r $, $ s $ and $ G $ are increasing.
		
	\end{enumerate}
\end{proof}



\end{document}